\documentclass[letterpaper, 10 pt, conference]{ieeeconf}  

\IEEEoverridecommandlockouts                              
\usepackage{amsmath,amssymb,amsfonts}

\usepackage{algorithm}
\usepackage{array}
\usepackage[caption=false,font=normalsize,labelfont=sf,textfont=sf]{subfig}
\usepackage{textcomp}
\usepackage{stfloats}
\usepackage{url}
\usepackage{verbatim}
\usepackage{graphicx}
\usepackage{cite}
\usepackage{fancyhdr}

\usepackage{algorithm}
\usepackage{algpseudocode}
\algnewcommand\algorithmicinput{\textbf{Input:}}
\algnewcommand\Input{\item[\algorithmicinput]}
\algnewcommand\algorithmicoutput{\textbf{Output:}}
\algnewcommand\Output{\item[\algorithmicoutput]}
\usepackage{graphicx}
\newtheorem{lemma}{Lemma}

\newtheorem{thm}{Theorem}
\newtheorem{remark}{Remark}

\newtheorem{prop}{Proposition}
\newtheorem{definition}{Definition}
\newtheorem{corollary}{Corollary}
\usepackage{comment}
\usepackage{cite}
\usepackage{xcolor}

\usepackage{lipsum}
\usepackage{multicol}
\usepackage{graphicx}

\title{\LARGE \bf
Construction of the Sparsest Maximally $r$-Robust Graphs
}

\author{Haejoon Lee$^{1}$ and Dimitra Panagou$^{1,2}$
\thanks{*This work was supported by the Air Force Office of Scientific Research (AFOSR) under Award No. FA9550-23-1-0163.}
\thanks{*This work was partially sponsored by the Office of Naval Research (ONR), under grant number N00014-20-1-2395. The views and conclusions contained herein are those of the authors only and should not be interpreted as representing those of ONR, the U.S. Navy or the U.S. Government.}
\thanks{*Emails: $\{$haejoonl, dpanagou$\}$@umich.edu. $^{1}$Department of Robotics,
        University of Michigan, Ann Arbor, MI, USA. $^{2}$Department of Aerospace Engineering,
        University of Michigan, Ann Arbor, MI, USA}%
}

\begin{document}
\maketitle
\pagestyle{empty}

\begin{abstract} In recent years, the notion of $r$-robustness for the communication graph of the network has been introduced to address the challenge of achieving consensus in the presence of misbehaving agents. Higher $r$-robustness typically implies higher tolerance to malicious information towards achieving resilient consensus, but it also implies more edges for the communication graph. This in turn conflicts with the need to minimize communication due to limited resources in real-world applications (e.g., multi-robot networks). In this paper, our contributions are twofold. (a) We provide the necessary subgraph structures and tight lower bounds on the number of edges required for graphs with a given number of nodes to achieve maximum robustness. (b) We then use the results of (a) to introduce two classes of graphs that maintain maximum robustness with the least number of edges. Our work is validated through a series of simulations.
\end{abstract}

\section{Introduction}
Distributed multi-robot systems are deployed for various tasks such as information gathering \cite{viseras2018}, target tracking \cite{zhang2021}, and collaborative manipulation \cite{yanhao2021}. One way to achieve these different tasks is through the consensus algorithm, in which multiple robots achieve an agreement to common state values. However, consensus algorithm performance deteriorates significantly when one or more compromised robots share wrong, or even adversarial, information. 

To address this issue, there have been many works on providing resilience to consensus \cite{LeBlanc13,yan2022,zhang2012,leblanc2013alg,zhang2015,usevitch2020FiniteTime}. 
In \cite{LeBlanc13,zhang2012,zhang2015}, an algorithm called \textit{Weighted Mean-Subsequence-Reduced} (W-MSR) was introduced to allow the non-compromised (often called normal) agents to reach consensus despite the presence of compromised robots. Under the graph topological property called $\textit{r-robustness}$, the W-MSR algorithm guarantees that, under a certain level of robustness of the communication graph of the network, and for a given upper bounded number of compromised agents, normal agents can reach consensus within the convex hull of their initial values to successfully complete the given tasks despite the compromised agents.

By definition, a communication graph of the network needs more edges to achieve higher $r$-robustness. Nevertheless, given practical challenges such as limited communication range, bandwidth, and energy, it might be beneficial for the multi-robot network to minimize communications as much as possible. Therefore, we are interested in finding graphs of maximum robustness with the least number of edges. In other words, we are interested in finding the sparsest $r$-robust graph topologies with maximum robustness. In this paper, we introduce two classes of $r$-robust graphs that show such properties with a given number of nodes.

The maximum $r$-robustness for a given $n$-vertex network graph is $\lceil\frac n 2\rceil$ \cite{LeBlanc13}. One trivial example of a graph with the maximum robustness is a complete graph, but such graph uses the maximum number of edges possible. In fact, the relation between a graph's $r$-robustness and number of edges has not been explored in detail. A lower bound of the number of edges for undirected $(2,2)$-robust graphs is given in \cite{saldana2016}, but this bound does not hold for graphs of other robustness levels (i.e., for other than $(2,2)$-robust graphs). A preferential-attachment method is investigated in \cite{zhang2012,LeBlanc13, zhang2015} to systematically increase the size of a graph while maintaining its robustness. In \cite{zhang2012,LeBlanc13}, the minimum degree needed for any $r$- and $(r,s)$-robust graphs is presented. Nevertheless, these bounds are local in the sense that they only apply to the individual nodes. In this paper, we are interested in obtaining lower bounds on the number of edges of whole graphs with varying $r$-robustness.

 The authors in \cite{Guerrero17} study the construction of a class of undirected $r$-robust graphs that uses minimum number of nodes to achieve the maximum robustness. However, their construction method is not concerned with minimizing the usages of edges. Furthermore, several works study robustness of 2D lattice-based geometric graphs for systematic expansions of robotic networks \cite{Guerrero20, Guerrero19}. While these works expand the network graphs with predetermined robustness, their graphs do not minimize the number of edges. Conversely, our work constructs classes of graphs that specifically utilize the least number of edges for the maximum robustness.


Some approaches study how to maintain, or increase, the $r$-robustness of a graph by controlling its algebraic connectivity $\lambda_2$, i.e., the second smallest eigenvalue of a graph's Laplacian matrix \cite{cavorsi2022,saulnier2017}, using the bound $r\geq \lceil{\frac {\lambda_2} {2}}\rceil$ introduced in \cite{saulnier2017}. However, this approach mainly aims to increase the connectivity of the graph without considering its robustness directly. This leads to graphs that do not necessarily maintain the smallest number of edges as the robustness of the graph increases. In contrast, our work introduces fundamental structures the graphs should have for certain robustness and thus offers insights in exactly which edges are needed to increase its robustness.

\emph{Contributions:} The contributions of this paper are twofold. (a) We first present necessary subgraph structures and tight lower bounds on the number of edges graphs need to achieve maximum robustness. (b) We introduce two classes of graphs that use the least number of edges to attain maximum robustness. We also present simulations to verify their robustness and properties of being the sparsest.

\emph{Organization}: In Section~\ref{sec:prelim}, we go over the notations as well as fundamental concepts of the W-MSR algorithm and $r$-robustness. We provide necessary subgraph structures and lower bounds on the number of edges for $r$-robust graphs with a given number of nodes in Section~\ref{sec:lower}. In Section~\ref{sec:construction}, we introduce and provide systematic constructions of the graphs that use the least number of edges to maximize their robustness levels with given numbers of nodes. In Section~\ref{sec:sim}, we present our simulation results, and in Section~\ref{sec:concl}, we present our conclusions and outline avenues of future work.


\section{Preliminaries}
\label{sec:prelim}

\subsubsection{Notation and Basic Graph Theory}
We denote a simple, undirected time-invariant graph as $\mathcal G = (\mathcal {V},\mathcal{E})$ where $\mathcal {V}$ and $\mathcal{E}$ are the vertex set and the undirected edge set of the graph respectively. A undirected edge $(i,j)\in \mathcal{E}$ indicates that information is exchanged between the nodes $i$ and $j$. The neighbor set of agent $i$ is denoted as $\mathcal V_i^N =\{j\in \mathcal V \;|\; (i,j) \in \mathcal{E}\}$. The state of agent $i$ at time $t$ is denoted as $x_i[t]$. We denote the cardinality of a set $\mathcal S$ as $|\mathcal S|$. We denote the set of non-negative and positive integers as $\mathbb Z_{\geq 0}$ and $\mathbb Z_{>0}$.

\subsubsection{Fundamentals of W-MSR and $r$-Robustness}
The notions of W-MSR algorithm and $r$-robustness of a graph are defined in \cite{LeBlanc13}. In this section, we review some relevant fundamental concepts.

Let $\mathcal G = (\mathcal V, \mathcal E)$ be a graph. Then, a robot $i \in \mathcal V$ shares its state $x_i[t]$ at time $t$ with all neighbors $j \in \mathcal V_i^N$, and each robot $i$ updates its state according to the nominal update rule:
\begin{equation}
\label{eq:linear}
    x_i[t+1]=w_{ii}[t]x_i[t] + \sum_{j\in \mathcal V_i^N} w_{ij}[t]x_j[t],
\end{equation}
where $w_{ij}[t]$ is the weight assigned to agent $j$'s value by agent $i$, and where the following conditions are assumed for $w_{ij}[t]$ for $\forall i \in \mathcal V$ and $t \in \mathbb Z_{\geq0}$:
\begin{itemize}
    \item $w_{ij}[t]=0$ $\forall j\notin \mathcal V_i^N\cup \{i\}$,
    \item $w_{ij}[t]\geq \alpha$, $0\leq \alpha <1$ $\forall j\in \mathcal V_i^N\cup \{i\}$,
    \item $\sum_{j=1}^n w_{ij}[t]=1$
\end{itemize}
Through the protocol given by \eqref{eq:linear}, agents reach asymptotic consensus as long as the graph has a \textit{rooted-out branching} (i.e. there exists a node that has paths to all other nodes in the graph) \cite{mesbahi2010book}. However, this algorithm loses its consensus guarantee in the presence of misbehaving agents, whose formal definitions are given below:
\begin{definition}[\textbf{misbehaving agent}]
    \label{misbeh}
    An agent is \textbf{misbehaving} if it does not follow the nominal update protocol \eqref{eq:linear} at some time stamp $t$.
\end{definition}
There are numerous scopes of threat that describe the number of misbehaving agents in a network \cite{LeBlanc13}. The misbehaving agents are assumed to adopt one scope, and we discuss two of them below:
\begin{definition}[$\mathbf F$-\textbf{total}]
    \label{ftotal}
    A set $\mathcal S \subset \mathcal V$ is $\mathbf F$-\textbf{total} if it contains at most $F$ nodes in the graph (i.e. $|\mathcal S|\leq F)$.
\end{definition}

\begin{definition}[$\mathbf F$-\textbf{local}]
    \label{ftotal}
    A set $\mathcal S \subset \mathcal V$ is $\mathbf F$-\textbf{local} if all other nodes have at most $F$ nodes of $\mathcal S$ as their neighbors (i.e. $|\mathcal V_i^N\cap \mathcal S|\leq F$, $\forall i \in \mathcal V \setminus \mathcal S$).
\end{definition}

In response, algorithms on resilient consensus \cite{LeBlanc13, saldana2017,dibaji2017,yan2022} have become very popular. In particular, the W-MSR (Weighted-Mean Subsequent Reduced) algorithm \cite{LeBlanc13} with the parameter $F$ guarantees normal agents to achieve asymptotic consensus on their state values with either $F$-total or $F$-local misbehaving agents under certain assumed topological properties of the communication graph. We review these properties and other relevant concepts below:
\begin{definition}[$\mathbf r$-\textbf{reachable}]
    \label{reachability}
    Let $\mathcal G = (\mathcal V,\mathcal{E})$ be a graph and $S$ be a nonempty subset of $\mathcal V$. The subset $S$ is $\mathbf r$-\textbf{reachable} if $\exists i\in S$ such that $|\mathcal V^N_i \backslash S|\geq r$, where $r\in \mathbb Z_{\geq 0}$.
\end{definition}

\begin{definition}[$\mathbf r$-\textbf{robust}]
    \label{robustness}
    A graph $\mathcal G = (\mathcal V,\mathcal{E})$ is $\mathbf r$-\textbf{robust} if $\forall S_1,S_2 \subset \mathcal V$ where $S_1\cap S_2 = \emptyset$ and $S_1,S_2\neq \emptyset$, at least one of them is $r$-reachable.
\end{definition}

The W-MSR algorithm and $r$-robustness are closely related. The W-MSR algorithm allows normal agents in a network to reach consensus on a value within convex hull of their initial values, and $r$-robustness dictates the number of misbehaving agents the algorithm can tolerate while still having a consensus guarantee. A network being $(2F+1)$-robust is a sufficient condition for its normal agents to reach consensus through the W-MSR algorithm in the presence of $F$-local or $F$-total misbehaving agents \cite{LeBlanc13}.

\begin{figure*}[ht]
\centering
\includegraphics[width=\textwidth]{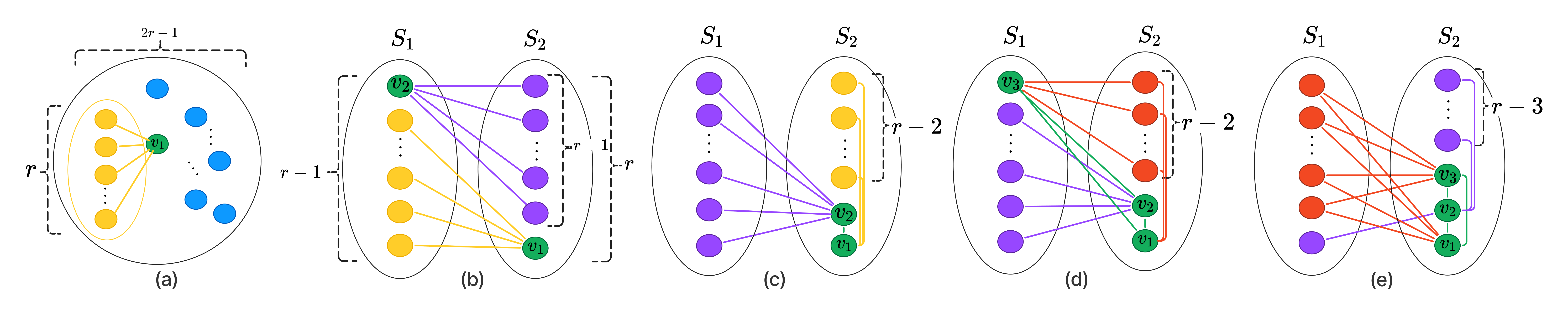}
  \caption{This shows the snapshots of the initial constructions of $S_1$ and $S_2$ from start of the first to end of the second step in the proof of Lemma~\ref{lem:complete}. Figure (a) shows that any node $v_1$ (colored in green) will have $r$ neighbors (colored in yellow). Note that Figure (b) indicates $|S_1|=r-1$ and $|S_2|=r$, which are fixed for the figures (c), (d), and (e). Figure (b) shows the start of the first step where a node $v_2 \in S_1$ is connected to all the $r-1$ purple nodes and $v_1$ in $S_2$. The green nodes $v_1$ and $v_2$ together form a $2$-clique. Figure (c) shows $S_1$ and $S_2$ of the end of the first step after we swap $r-1$ nodes including $v_2$ in $S_1$ with $r-1$ purple nodes in $S_2\setminus\{v_1\}$. Then, the start of the second step is shown in Figure (d), where another node $v_3 \in S_1$ has edges with $r-2$ red nodes and $v_1,v_2$ in $S_2$. Figure (e) shows $S_1$ and $S_2$ after we swap $r-2$ nodes in $S_1$ including $v_3$ with $r-2$ red nodes in $S_2\setminus\{v_1,v_2\}$. At the end of second step, we have a $3$-clique (colored in green in Figure (e)). This process continues until $(r+1)$-clique is formed.}
  \label{fig:complete_graph_proof}
\end{figure*}

\section{Lower Bound of Number of Edges for $r$-robust Graphs}
\label{sec:lower}
Our goal is to find classes of graphs that have the maximum $r$-robustness with the least number of edges. Since the maximum robustness any graph of $n$ nodes can achieve is $\lceil\frac n 2\rceil$, we consider graphs with either $2r-1$ (for odd values of $n$) or $2r$ (for even values of  $n$) nodes. Therefore in this section, we aim to find the tight lower bounds of number of edges that $r$-robust graphs need to satisfy, for both cases of $2r-1$ and of $2r$ nodes. These are formally established in Theorem~\ref{thm:min_addition} and Theorem~\ref{thm:min_addition2}, respectively. We first introduce the concept of clique:
\begin{definition}[$\textbf {clique}$ \cite{clique}]
    A $\textbf {clique}$ $C$ of a graph $\mathcal G = (\mathcal V, \mathcal{E})$ is a subset of $\mathcal V$ whose elements are adjacent to each other in $\mathcal G$.
\end{definition}
In general, $k$-clique refers to a clique size of $k$ nodes. The notion of a clique plays a vital role in building the structures of $r$-robust graphs, easing the difficulty in proving the lower bounds on the number of edges in a general graph setting. Now, we show the maximum clique any $r$-robust graphs with $2r-1$ nodes should contain in Lemma~\ref{lem:complete}.

\subsection{Case 1: $r$-Robust Graph with $2r-1$ Nodes}
\begin{figure*}[h]
\centering
\includegraphics[width=\textwidth]{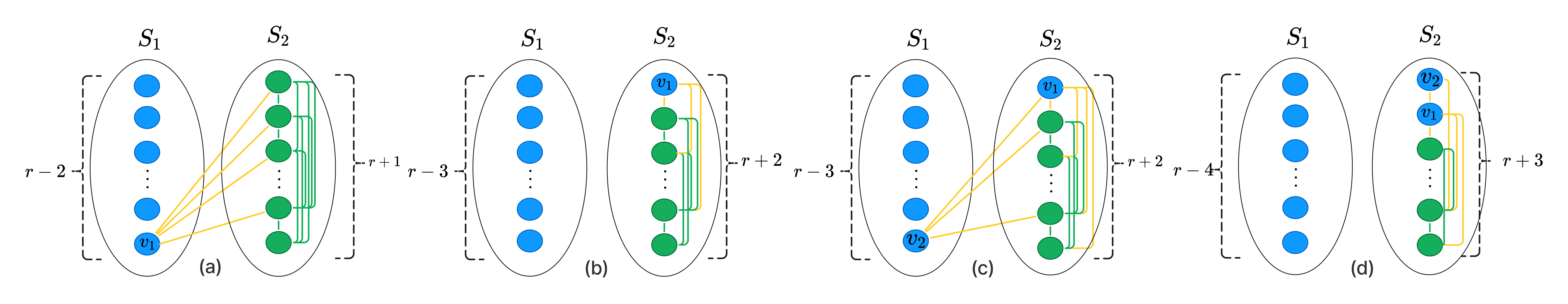}
  \caption{This presents the snapshots of initial developments of $S_1$ and $S_2$ from the start of the first step to the end of the second step in the proof for Theorem~\ref{thm:min_addition}. The first step starts with $|S_1|=r-2$ and $|S_2|=r+1$, where $S_2$ contains a $(r+1)$-clique from Lemma~\ref{lem:complete} whose nodes are colored in green. Then, since $S_1$ is $r$-reachable, a node $v_1 \in S_1$ must have edges (colored in yellow) with any $r$ nodes in $S_2$. This requires $r$ additional edges. However, once $v_1 \in S_1$ is moved to $S_2$ (shown Figure (b)), at second step, it has a node $v_2$ that has edges (colored in yellow) with any $r$ nodes in $S_2$. This is because $S_1$ is $r$-reachable even without $v_1$. This requires $r$ additional edges again, as shown on Figure (c). Lastly, the second step ends after $v_2 \in S_1$ is moved to $S_2$, which is shown on Figure (d). After two steps, we needed $2r$ additional edges. This process continues until $S_1$ becomes empty.}
  \label{fig:thm1}
  \vspace{-3mm}
\end{figure*}

\begin{lemma}
\label{lem:complete}
    Let $\mathcal G = (\mathcal V, \mathcal{E})$ be an $r$-robust graph with $|\mathcal V|=2r-1$. Then $\mathcal V$ must contain a $(r+1)$-clique.
\end{lemma}

\begin{proof}
    In this proof, we will show our argument by constructing two subsets $S_1$ and $S_2$ of $\mathcal V$ such that the $r$-robust graph $\mathcal G$ contains a $(r+1)$-clique. Since $\mathcal G$ is $r$-robust, it holds that for any pair of disjoint, non-empty subsets $S_1$, $S_2$ of $\mathcal V$, at least one of them is $r$-reachable. Hence, every node $v_i\in \mathcal V$ has a degree of $r$. WLOG, let $v_1$ be a node as shown on Fig.~\ref{fig:complete_graph_proof} (a). Then we construct $S_1$ to be the set of $r-1$ neighbors of $v_1$ and $S_2$ to be the set of the remaining $r$ nodes including $v_1$ such that $|S_1|=r-1$, $|S_2|=r$, and $S_1\cap S_2=\emptyset$. This enforces $S_1$ to be $r$-reachable.

    Since $S_1$ is $r$-reachable, there exists a node $v_2 \in S_1$ that has edges with all $r$ nodes in $S_2$. Then in the first step, $v_2$ has an edge with $v_1 \in S_2$, forming a $2$-clique as shown on Fig.~\ref{fig:complete_graph_proof} (b). We then swap $r-1$ nodes in $S_1$ including $v_2$ with $r-1$ nodes in $S_2\setminus\{v_1\}$ (colored in purple in Fig.~\ref{fig:complete_graph_proof}). Then in the second step, since $S_1$ is $r$-reachable even after $v_2$ is swapped out, it has a node $v_3 \in S_1$ that has edges with $r$ nodes in $S_2$. By construction, since $v_1$ and $v_2$ are among the $r$ nodes in $S_2$, $v_3$ forms a $3$-clique with $v_1$ and $v_2$, as shown on Fig.~\ref{fig:complete_graph_proof} (d). We then swap $r-2$ nodes in $S_1$ including $v_3$ with $r-2$ nodes in $S_2\setminus \{v_1,v_2\}$ (colored in red in Fig.~\ref{fig:complete_graph_proof}). Again, in the third step, since $S_1$ is $r$-reachable even after $v_3$ is swapped out, it has a node $v_4$ that has edges with all $r$ nodes in $S_2$. Since $v_1,v_2,v_3 \in S_2$ are three of the 
    $r$ nodes to have edges with $v_4$, the nodes $v_1,v_2,v_3,v_4$ form a $4$-clique. We then swap $r-3$ nodes in $S_1$ including $v_4$ with $r-3$ nodes in $S_2\setminus \{v_1,v_2,v_3\}$, and thus the construction continues.
    
    Likewise, in the $k^{\text{th}}$ step, a node $v_{k+1}$ (that has edges with $r$ nodes in $S_2$ in the $k^{\text{th}}$ step) forms a $(k+1)$-clique with $v_1,v_2,\cdots, v_{k-1},v_k \in S_2$. We then swap $r-k$ nodes in $S_1$ including $v_{k+1}$ with $r-k$ nodes in $S_2\setminus \{v_1,\cdots,v_k\}$. At the end of $r-1^{\text{th}}$ step, $S_2$ contains $v_1,\cdots,v_r$ that form an $r$-clique. Finally, since $S_1$ is $r$-reachable, it contains a node $v_{r+1}$ that has edges with all nodes $v_1,\cdots, v_r \in S_2$. Then, $\{v_{r+1}\}\cup S_2$ forms a ($r+1)$-clique.
\end{proof}

\begin{thm}
 \label{thm:min_addition}
    Let $\mathcal G = (\mathcal V, \mathcal{E})$ be an $r$-robust graph with 
$|\mathcal V|= 2r-1$. Then, 
\begin{equation}
        \label{eq:bound1}
        |\mathcal{E}|\geq \frac {3r(r-1)} {2}.
    \end{equation}
\end{thm}

\begin{proof}
 Since $\mathcal G$ is $r$-robust, it holds that for any pair of disjoint, non-empty subsets $S_1$, $S_2$ of $\mathcal V$, at least one of them is $r$-reachable.
 From Lemma~\ref{lem:complete}, we know that $\mathcal V$ contains a $(r+1)$-clique $C$. Let $C' =  \mathcal V\setminus C$. Then, $|C|= r+1$ and $|C'|=r-2$. Since $C \in \mathcal V$, we initially have $|\mathcal E| \geq \frac {r(r+1)} 2$.
 
  WLOG, let $S_1=C'$ and $S_2=C$, where $|S_1|=r-2$ and $|S_2|=r+1$. This implies $S_1$ is $r$-reachable. Then in the first step, since $S_1$ is $r$-reachable, it contains one node $v_1$ that has edges with $r$ nodes in $S_2$. This pair of $S_1$ and $S_2$ is visualized in Fig.~\ref{fig:thm1} (a). Note that all nodes in $C$ are colored in green in Fig.~\ref{fig:thm1}. Since $v_1$ has $r$ edges, it requires $r$ additional edges (i.e. $|\mathcal E| \geq \frac {r(r+1)} 2+r$). We then move $v_1$ from $S_1$ to $S_2$, as shown on Fig.~\ref{fig:thm1} (b). In the second step, since $S_1$ is $r$-reachable, it contains a node $v_2$ that has edges with at least $r$ nodes in $S_2$ (shown on Fig.~\ref{fig:thm1} (c)). This again requires at least $r$ additional edges (i.e. $|\mathcal E| \geq \frac {r(r+1)} 2+2r$). We then move $v_2$ from $S_1$ to $S_2$, as shown in Fig.~\ref{fig:thm1} (d).
    
    We can continue this process of (1) drawing edges between a node $v_i \in S_1$ and any $r$ nodes in $S_2$ and (2) putting $v_i$ into $S_2$ in the $i^{\text{th}}$ step until $S_1$ becomes empty. Then, we can do it for $r-2$ times, as we initially have $|S_1|=r-2$. In other words, we have $r-2$ different pairs of $S_1$ and $S_2$, and each pair requires $r$ new edges. Since we started from the setting where $S_2=C$ and forced $|\mathcal E|$ to increase at least a total of $r(r-1)$, $|\mathcal E| \geq \frac {r(r+1) + 2r(r-2)} 2 \Rightarrow |\mathcal E| \geq \frac {3r(r-1)} 2$. 
\end{proof}

From Theorem~\ref{thm:min_addition}, we can get the following corollary:
\begin{corollary}
 \label{cor:lowerbound}
    Let $\mathcal G = (\mathcal V, \mathcal{E})$ be an $r$-robust graph. Then, \begin{equation}
        \label{eq:bound1}
        |\mathcal{E}|\geq \frac {3r(r-1)} {2}.
    \end{equation}
\end{corollary}
\begin{proof}
In general, the more nodes a graph has, the
more edges it must have to maintain a certain robustness level. Since (i) $2r-1$ nodes is the least number of nodes a graph has to have to be $r$-robust \cite{Guerrero17} and (ii) a graph of $2r-1$ nodes needs at least $\frac {3r(r-1)} {2}$ edges to be $r$-robust from Theorem~\ref{thm:min_addition}, any $r$-robust graph $\mathcal G$ needs to have at least $\frac {3r(r-1)} {2}$ edges, regardless of its number of nodes. 
\end{proof}

 Note that Corollary~\ref{cor:lowerbound} illustrates that the lower bound $ |\mathcal{E}|\geq \frac {3r(r-1)} {2}$ from Theorem~\ref{thm:min_addition} is in fact a necessary condition any $r$-robust graphs have to satisfy regardless of its number of nodes. This is important, because this guarantees that graphs with edges less than $\frac {3r(r-1)} 2$ cannot be $r$-robust.

\subsection{Case 2: $r$-Robust Graph with $2r$ Nodes}

We now study the lower bound of edges for $r$-robust graphs with $2r$ nodes.

\begin{figure*}[h]
\centering
\includegraphics[width=\textwidth]{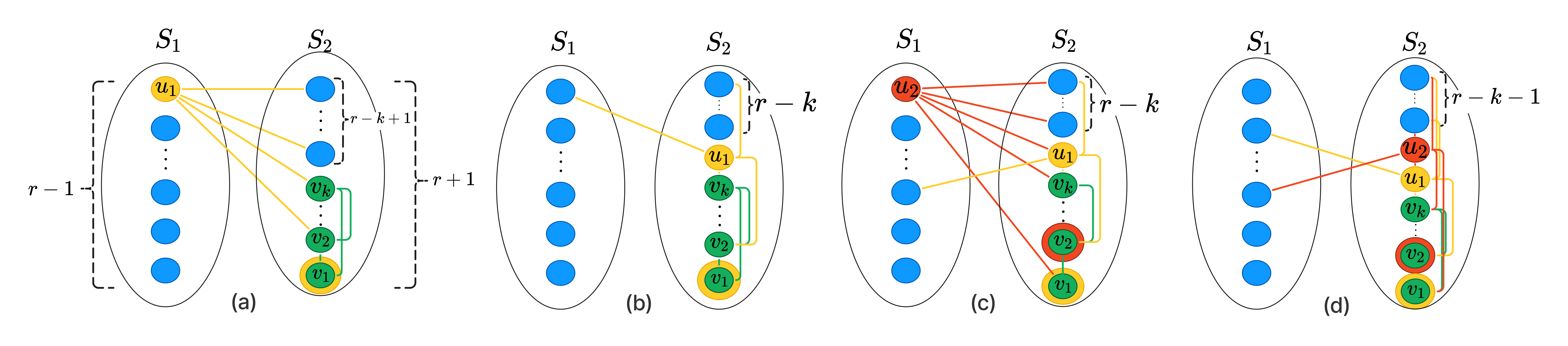}
  \caption{This presents the snapshots of two stages of inductive developments of $S_1$ and $S_2$ in the proof for Lemma~\ref{lem:complete2}. Figure (a) indicates $|S_1|=r-1$ and $|S_2|=r+1$, which are fixed, and thus omitted in the other figures. 
  In Figure (a), $S_2$ contains $r+1$ nodes including $k$-clique $C_{k,1}=\{v_1,\cdots,v_k\}$ which are colored in green, while $S_1$ contains the remaining $r-1$ nodes. In the first stage, since $S_1$ is $r$-reachable, a node $u_1 \in S_1$ has edges with at least $r$ nodes in $S_2$. Figure (a) shows the third case scenario where the yellow node $u_1$ does not have an edge with $v_1$, which is highlighted in yellow. Then, $u_1$ forms a $k$-clique $C_{k,2}$ with $C_{k,1}\setminus\{v_1\}$. That implies that if any node $u_i \in S_1$, $i\neq 1$, has edges with at least $r$ nodes in $S_2$ but not with $v_1$ in any future steps, $u_i$ will form a $(k+1)$-clique with $C_{k,2}$. We then swap $u_1$ with a node in $S_2\setminus C_{k,1}$, as shown in Figure (b). Then in the second stage, since $S_1$ is $r$-reachable, a node $u_2 \in S_2$ (colored in red in the figure) has edges with at least $r$ nodes in $S_2$. Figure (c) presents the third case scenario where $u_2$ only has $r$ edges and does not have an edge with $v_2$, which is highlighted in red. Similar to before, if any $u_i \in S_2$, $i\neq 2$, in future steps has edges with $r$ nodes except for $v_2$ again, that forms a $(k+1)$-clique. Then, we swap $u_2$ with a node in $S_2\setminus U'$ where $U'=C_{k,1}\cup C_{k,2}$, as shown on Figure (d). Still, $S_1$ is $r$-reachable, and thus the process continues until a $\big(\lfloor\frac {r+3}{2}\rfloor\big)$-clique is formed.}
  \label{fig:complete2}
  \vspace{-4mm}
\end{figure*}
\begin{lemma}
\label{lem:complete2}
    Let $\mathcal G = (\mathcal V, \mathcal{E})$ be an $r$-robust graph with $|\mathcal V| =2r$. Then $\mathcal V$ must contain a $\big(\lfloor\frac {r+4}{2}\rfloor\big)$-clique.
\end{lemma}

\begin{proof}
In this proof, we will inductively show our argument by constructing two subsets $S_1$ and $S_2$ of $\mathcal V$ such that $\mathcal G$ is $r$-robust and contains a $c$-clique, $2 \leq c \leq \lfloor\frac {r+4}{2}\rfloor$. Since $\mathcal G$ is $r$-robust, it holds that for any pair of disjoint, non-empty subsets $S_1$, $S_2$ of $\mathcal V$, at least one of them is $r$-reachable. Hence we have the base case of a $2$-clique.

For the induction step, we fix $|S_1|=r-1$ and $|S_2|=r+1$. This enforces $S_1$ to be $r$-reachable. WLOG, let $S_2$ contain $r+1$ nodes including $k$-clique $C_{k,1}=\{v_1,\cdots,v_k\}$ (green in Fig.~\ref{fig:complete2}), where $k\geq 2$, and let $S_1$ contain the remaining $r-1$ nodes. We also denote $U_1 = U_1' = C_{k,1}$. Since $S_1$ is $r$-reachable, it has a node $u_1$ (yellow in Fig.~\ref{fig:complete2}) that has edges with at least $r$ nodes in $S_2$. There are three cases: $u_1$ has edges with (i) all $r+1$ nodes in $S_2$, (ii) $r$ nodes but not with a node in $S_2\setminus U_1$, or (iii) $r$ nodes but not with $v_1 \in U_1$ (highlighted in yellow in Fig.~\ref{fig:complete2}). If (i) or (ii), $u_1$ forms a $(k+1)$-clique with $C_{k,1}$. If (iii), $u_1$ forms another $k$-clique $C_{k,2}$ with $U_1' \setminus\{v_1\}$. Note that $u_1$ has edges with all nodes in $S_2\setminus\{v_1\}$. Now, if (i) or (ii), we have a $(k+1)$-clique and done. If (iii), let $U_2= C_{k,1} \cap C_{k,2} = \{v_2,\cdots,v_k\}$ and $U_2'=C_{k,1} \cup C_{k,2} =\{v_1,\cdots,v_k,u_1\}$. We also swap $u_1$ with any of the nodes in $S_2\setminus U_1'$ to get $S_1$ and $S_2$ shown on Fig.~\ref{fig:complete2} (b). Then, since $S_1$ is $r$-reachable, another node $u_2 \in S_1$ (colored in red in Fig.~\ref{fig:complete2}) has edges with at least $r$ nodes in $S_2$. Again, there are three cases: $u_2$ has edges with (i) all $r+1$ nodes in $S_2$, (ii) $r$ nodes but not with a node in $S_2\setminus U_2$, or (iii) $r$ nodes but not with $v_2 \in U_2$ (highlighted in red in Fig.~\ref{fig:complete2}). If (i), $u_2$ forms a $(k+1)$-clique with $C_{k,j}$ for any $j\in \{1,2\}$. If (ii), there are $3$ subcases: $u_2$ does not have an edge only with (1) $u_1$, (2) $v_1$, or (3) any other node in $S_2\setminus U_2$. Then, $u_2$ forms a $(k+1)$-clique with $C_{k,1}$ if (1), $C_{k,2}$ if (2), or $C_{k,j}$ for any $j\in \{1,2\}$ if (3). If (iii), $u_2$ forms another $k$-clique $C_{k,3}$ with $U_2'\setminus\{v_1,v_2\}=\{v_3,\cdots,v_k,u_1\}$. Also note that $u_2$ has edges with all nodes in $S_2\setminus \{v_2\}$ including $u_1$. Now, if (i) or (ii), we have a $(k+1)$-clique and done. If (iii), we update $U_3=\bigcap\limits_{j=1}^3 C_{k,j}=\{v_3,\cdots,v_k\}$ and $U_3'=\bigcup\limits_{j=1}^3 C_{k,j}=\{v_1,\cdots,v_k,u_1,u_2\}$. We also swap $u_2$ with any node in $S_2\setminus U_2'$. Again, since $S_1$ is $r$-reachable, there is a node $v_3 \in S_1$ that has edges with at least $r$ nodes in $S_2$. There are three cases: $u_3$ has edges with (i) all $r+1$ nodes in $S_2$, (ii) $r$ nodes but not with a node in $S_2\setminus U_3$, or (iii) $r$ nodes but not with $v_3 \in U_3$. If (i), $u_3$ forms a $(k+1)$-clique with $C_{k,j}$ for any $j\in \{1,2,3\}$. If (ii), there are $4$ subcases: $u_3$ does not have an edge only with (1) one of $\{u_1,u_2\}$, (2) $v_1$, (3) $v_2$, or (4) any other node in $S_2\setminus U_2$. Then, $u_3$ forms a $(k+1)$-clique with $C_{k,1}$ if (1), $C_{k,2}$ if (2), $C_{k,3}$ if (3), and $C_{k,j}$ for any $j\in \{1,2,3\}$ if (4). If (iii), $u_3$ forms another $k$-clique $C_{k,4}$ with $U_3'\setminus\{v_1,v_2,v_3\}=\{v_4,\cdots,v_k,u_1,u_2\}$. Also note that $u_3$ has edges with all nodes in $S_2\setminus\{v_3\}$ including $u_1,u_2$. Now, if (i) or (ii), we have a $(k+1)$-clique and done. If (iii), we swap $u_3$ with a node in $S_2\setminus U_3'$, continuing the process.
    
    Likewise, unless we get (i) or (ii), we always get three cases as illustrated in the previous steps. Let $u_n \in S_1$ have edges with at least $r$ nodes in $S_2$, $n\in\{1,\cdots,k\}$. There are three cases: $u_n$ has edges with (i) all $r+1$ nodes in $S_2$, (ii) $r$ nodes but not with a node in $S_2\setminus U_n$ where $U_n=\bigcap\limits_{j=1}^{n}C_{k,j}=\{v_n,\cdots,v_k\}$, or (iii) $r$ nodes but not with $v_n \in U_n$. If (i), $u_n$ forms a $(k+1)$-clique with $C_{k,j}$ for any $j\in \{1,\cdots, n\}$. If (ii), there are $n+1$ subcases: $u_n$ does not have an edge only with (1) one of $\{u_1\cdots u_{n-1}$\}, (2) $v_1$, (3) $v_2$, $\cdots,$ ($n$) $v_{n-1}$, or ($n+1$) any other node in $S_2\setminus U_n$. Then, $u_n$ forms a $(k+1)$-clique with $C_{k,1}$ if (1), $C_{k,q}$ if ($q$) $\forall q \in \{2,3,\cdots,n-1,n$\}, or  $C_{k,j}$ for any $j \in \{1,\cdots, n\}$ if ($n+1$). If (iii), $u_n\in S_1$ forms another $k$-clique $C_{k,n+1}$ with $U_n'\setminus \{v_1,\cdots, v_n\}$ where $U_n'=\bigcup\limits_{j=1}^nC_{k,j}=\{v_1,\cdots,v_k,u_1,\cdots,u_{n-1}$\}. Now, if (i) or (ii), we have a $(k+1)$-clique and done. If (iii) we swap $u_n$ with any node in $S_2\setminus U_n'$ to continue the procedure. 
    
    If we continue encountering the third case for $k$ times, since $S_1$ is $r$-reachable, we will have $u_{k+1} \in S_1$ that has edges with at least $r$ nodes in $S_2$. Then, there are two cases: $u_{k+1}$ has edges with (1) all $r+1$ nodes in $S_2$, or (2) $r$ nodes but not with a node in $S_2\setminus U_{k+1}$ where $U_{k+1}=\bigcap\limits_{j=1}^{k+1}C_{k,j}=\emptyset$. In both cases, $u_{k+1}$ forms a $(k+1)$-clique with $C_{k,j}$ for any $j \in \{1,\cdots, k+1\}$, showing a $(k+1)$-clique must exist. Also note that $U_{k+1}'=\{v_1,\cdots,v_k,u_1,\cdots,u_{k}\} \subseteq S_2$. Since (i) at worst $u_1,\cdots, u_k$ as well as $v_1,\cdots,v_k$ need to be in $S_2$ for a $(k+1)$-clique to be formed and (ii) $|S_2|=r+1$, $k\leq \lfloor\frac {r+1}{2}\rfloor$. That means $\mathcal G$ contains a $\big(\lfloor\frac {r+3}{2}\rfloor\big)$-clique.

Continuing from the previous paragraph, at $k= \lfloor\frac {r+1}{2}\rfloor$ and for even values of $r$, at worst the third case mentioned above is repeated $\frac r 2$ times. Then, $S_2$ contains $K=\{u_1,\cdots, u_{\frac r 2}\}$, $C_{k,1}=\{v_1,\cdots, v_{\frac r 2}$\}, and a node $p$ not in $\big(\frac {r+2} 2\big)$-clique. Note that each of $K$ and $C_{k,1}$ forms a $\big(\frac {r} 2\big)$-clique. We know that a node $u_m \in S_1$, $m=\frac {r+2} 2$, must form a $m$-clique $C$ with either $K$ or $C_{k,1}$. We swap $u_m \in S_1$ with $p \in S_2$. Now let $P=S_2\setminus C$. Note that nodes in $P$ form a $(\frac r 2)$-clique such that $P\cap C=\emptyset$. Since $S_1$ is $r$-reachable even after $u_m$ is swapped out, it has a node $u_{m+1} \in S_1$ that has edges with at least $r$ nodes in $S_2$. There are three cases: $u_{m+1}$ has edges with (i) $r+1$ nodes in $S_2$, (ii) $r$ nodes but not with a node in $P$, or (iii) $r$ nodes but not with a node in $C$. If (i) or (ii), $u_{m+1}$ forms a $\big(\frac {r+4} 2\big)$-clique with all nodes in $C$. If (iii), $P\cup\{u_{m+1}\}$ form a $\big(\frac {r+2} 2\big)$-clique $P'$ such that $P'\cap C=\emptyset$. In this case, WLOG, let $P' \subset S_1$ and $C \subset S_2$ such that $|S_1|=|S_2|=r$. Since either $S_1$ or $S_2$ is $r$-reachable, all of the nodes in either $P'$ or $C$ have edges with one additional node in $S_2$ or $S_1$ respectively, forming a $\big(\frac {r+4} 2\big)$-clique. Thus, $\mathcal G$ must contain a $\big(\frac {r+4} 2\big)$-clique for even $r$. Since $\mathcal G$ must have a clique size of $\frac {r+3} 2$ and $\frac {r+4} 2$ for odd and even $r$ respectively, $\mathcal G$ must have a $\big(\lfloor{\frac {r+4} 2}\rfloor\big)$-clique.
\end{proof}


\begin{figure*}[h]
\centering
\includegraphics[width=\textwidth]{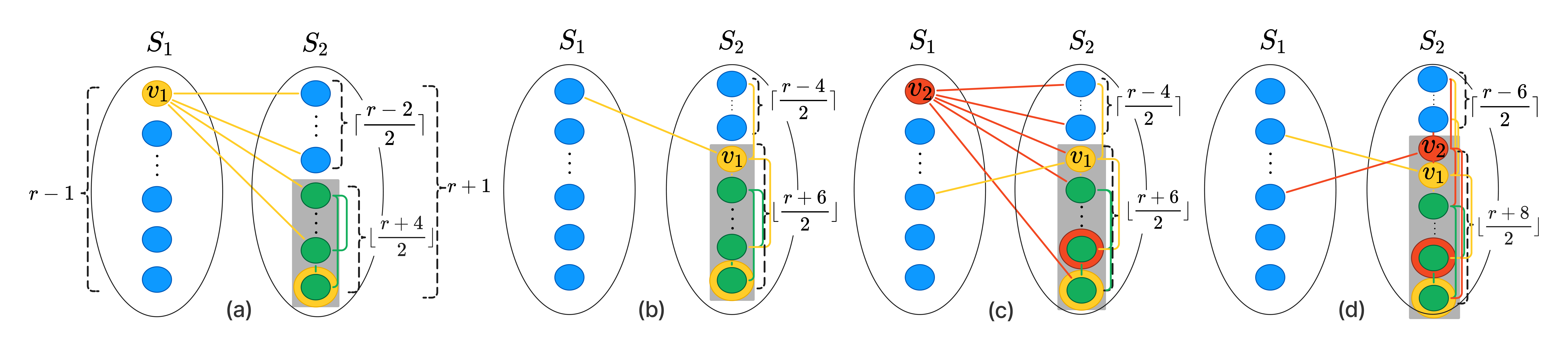}
  \caption{This presents the snapshots of two stages of inductive developments of $S_1$ and $S_2$ in the proof for Theorem~\ref{thm:min_addition2}. Figure (a) indicates $|S_1|=r-1$ and $|S_2|=r+1$, which are fixed, and thus omitted in the other figures. In Figure (a), $S_2$ contains $r+1$ nodes including $\big(\lfloor{\frac {r+4} 2}\rfloor\big)$-clique $C$ that are colored in green, while $S_1$ contains the remaining $r-1$ nodes. A gray box represents a vertex set $\mathcal V_S$ of $S$. Any node in the gray box is in $\mathcal V_S$. As $S_1$ is $r$-reachable, a node $v_1 \in S_1$ has edges with any $r$ nodes in $S_2$. Figure (a) shows the case where the yellow node $v_1$ does not have an edge with one node in $\mathcal V_S$, which is highlighted in yellow. We swap it with any node in $S_2\setminus \mathcal V_S$. Then, as $v_1$ becomes a vertex of subgraph $S$, as shown on Figure (b), $|\mathcal E_S|$ increases at least by $\lfloor{\frac {r+2} 2}\rfloor$. Since $S_1$ is $r$-reachable, a node $v_2 \in S_1$ (colored in red in the figure) has edges with at least $r$ nodes in $S_2$. Figure (c) presents the case where $v_2$ only has $r$ edges and does not have one with a node in $\mathcal V_S$, which is highlighted in red. Similar to before, we swap $v_2$ with a node in $S_2\setminus \mathcal V_S$. Then, as we add $v_2$ to the subgraph $S$, as shown on Figure (d), $|\mathcal E_S|$ increases at least by $\lfloor{\frac {r+4} 2}\rfloor$. This process continues until $S_2=\mathcal V_S$.}
  \label{fig:thm2}
  \vspace{-3mm}
\end{figure*}

Comparing Lemma~\ref{lem:complete} and~\ref{lem:complete2}, one can see that the size of a necessary maximum clique in an $r$-robust graph decreases as its number of node increases from $2r-1$ to $2r$. Using the results of Lemma~\ref{lem:complete2}, we further examine the necessary structure for $r$-robust graphs of $2r$ nodes.

\begin{lemma}
\label{lem:S}
       Let $\mathcal G = (\mathcal V, \mathcal{E})$ be an $r$-robust graph with $|\mathcal V|=2r$. Then, $\mathcal G$ contains an induced subgraph $S=(\mathcal V_S, \mathcal{E}_S)$ such that $|\mathcal V_S|=r+1$ and $|\mathcal{E}_S|\geq \lfloor{\frac {r^2+2} 2}\rfloor$.
\end{lemma}

\begin{proof} 
Since $\mathcal G$ is $r$-robust, for any pair of disjoint, non-empty subsets $S_1$, $S_2$ of $\mathcal V$, at least one of them is $r$-reachable. From Lemma~\ref{lem:complete2}, $\mathcal G$ contains a $\big(\lfloor\frac {r+4}{2}\rfloor\big)$-clique $C$.
    WLOG, let $S_2$ contain $r+1$ nodes including $C$ (green in Fig.~\ref{fig:thm2}), and let $S_1$ contain the remaining $r-1$ nodes. Initially we define $S=C$. Then $|\mathcal{E}_S|= \big(\frac 1 2 \big) \big(\lfloor\frac {r+4}{2}\rfloor\big) \big(\lfloor\frac {r+2}{2}\rfloor\big)$. Note that the nodes in $\mathcal V_S$ are contained in a gray box in Fig.~\ref{fig:thm2}. Since $|S_1|=r-1$ and $|S_2|=r+1$, $S_1$ must be $r$-reachable. 

    With Lemma~\ref{lem:complete2}, we have a $2$-clique and $3$-clique when $r=1$ and $r=2$ respectively, satisfying the statement in the lemma. Now, we consider $r\geq3$. With this setup, since $S_1$ is $r$-reachable, a node $v_1 \in S_1$ (yellow in Fig.~\ref{fig:thm2}) has edges with at least $r$ nodes in $S_2$. We swap $v_1 \in S_1$ with any node in $S_2\setminus \mathcal V_S$. Note that $|\mathcal V_S|=\lfloor\frac {r+4}{2}\rfloor$ at this point. Now, we add $v_1$ as a vertex of $S$ (i.e. $\mathcal V_S$ now contains $v_1$). Then, $|\mathcal E_S|$ increases at least by $\lfloor\frac {r+2}{2}\rfloor$, since $v_1$ can have no edge with at most one of the nodes in $\mathcal V_S$ (which is highlighted in yellow in Fig.~\ref{fig:thm2}). Since $S_1$ is $r$-reachable even without $v_1$, a node $v_2 \in S_1$ (colored in red in Fig.~\ref{fig:thm2}) has edges with at least $r$ nodes in $S_2$. Note that $|\mathcal V_S|=\lfloor\frac {r+6}{2}\rfloor$ at this point. We swap $v_2 \in S_1$ with a node in $S_2\setminus \mathcal V_S$ and add $v_2$ as a new vertex of $S$ (i.e. $\mathcal V_S$ now contains $v_2$), as shown on Fig.~\ref{fig:thm2} (d). This increases $|\mathcal E_S|$ at least by $\lfloor\frac {r+4}{2}\rfloor$, as $v_2$ could have no edge with at most one of the nodes in $\mathcal V_S$ (that is highlighted in red in Fig.~\ref{fig:thm2}). Then, since $S_1$ is still $r$-reachable, another node $v_3 \in S_1$ has edges with at least $r$ nodes in $S_2$. Note that $|\mathcal V_S|=\lfloor\frac {r+8}{2}\rfloor$. Again, we swap $v_3 \in S_1$ with a node in $S_2\setminus \mathcal V_S$ and add $v_3$ as a new vertex of $S$. Then, $|\mathcal E_S|$ increases at least by $\lfloor\frac {r+6}{2}\rfloor$. 
   
   This process continues until when $v_p \in S_1$, $p=r+1-\lfloor\frac {r+4}{2}\rfloor$, gets swapped with a node in $S_2\setminus \mathcal V_S$ where $\mathcal V_S=C\cup\{v_1,v_2,\cdots, v_{p-1}\}$. Here, adding $v_p$ as a new vertex of $S$ increases $|\mathcal E_S|$ at least by $r-1$. At this point, $S_2=\mathcal V_S$. Then, $|\mathcal V_S|=r+1$ and $|\mathcal E_S| \geq \big(\frac 1 2\big)\big(\lfloor\frac {r+4}{2}\rfloor\big) \big(\lfloor\frac {r+2}{2}\rfloor\big) + (\lfloor\frac {r+2}{2}\rfloor) + (\lfloor\frac {r+4}{2}\rfloor) + (\lfloor\frac {r+6}{2}\rfloor) +\cdots + (r-2) +(r-1)$.
    
    We remove the floor functions: if $r$ is even, $|\mathcal E_S| \geq \big(\frac 1 2\big)\big(\frac {r+4}{2}\big)\big(\frac {r+2} 2\big) +  \big(\frac {r+2} 2\big)+ \big(\frac {r+4}{2}\big) + \cdots + (r-2) +(r-1)= \frac {r^2+2} 2$. If $r$ is odd, $|\mathcal E_S| \geq \big(\frac 1 2\big)\big(\frac {r+3}{2}\big)\big(\frac {r+1} 2\big) +  \big(\frac {r+1} 2\big)+ \big(\frac {r+3}{2}\big) + \cdots + (r-2) +(r-1) = \frac {r^2+1} 2$. Thus, we have $|\mathcal V_S|=r+1$ and $|\mathcal E_S| \geq \lfloor{\frac{r^2+2} 2}\rfloor$.
\end{proof}

Now we use Lemma~\ref{lem:S} to show the lower bound on number of edges of an $r$-robust graph with $2r$ nodes.
\begin{thm}
        \label{thm:min_addition2}
    Let $\mathcal G = (\mathcal V, \mathcal{E})$ be an $r$-robust graph with $|\mathcal V|=2r$. Then, \begin{equation}
        \label{eq:bound2}
        |\mathcal{E}|\geq
            \left\lfloor{\frac {r(3r-2)+2} 2 }\right\rfloor.
    \end{equation}
\end{thm}

\begin{proof}
    Since $\mathcal G$ is $r$-robust, for any pair of disjoint, non-empty subsets $S_1$, $S_2$ of $\mathcal V$, at least one of them is $r$-reachable.
    From Lemma~\ref{lem:S}, $\mathcal G$ contains an induced subgraph $S=(\mathcal V_S, \mathcal E_S)$ such that $|\mathcal V_S|=r+1$ and $|\mathcal{E}_S|\geq \lfloor{\frac {r^2+2} 2}\rfloor$.

    Now, let $S_2=\mathcal V_S$ and $S_1=\mathcal V\setminus \mathcal V_S$, which means $|S_1|=r-1$ and $|S_2|=r+1$. This enforces $S_1$ to be $r$-reachable. Then, a node $v_1 \in S_1$ has edges with at least $r$ nodes in $S_2$. This requires at least $r$ edges (i.e. $|\mathcal E| \geq |\mathcal E_S|+r$). We move $v_1$ from $S_1$ to $S_2$. Then since $S_1$ is $r$-reachable even without $v_1$, a node $v_2 \in S_1$ has nodes with at least $r$ nodes in $S_2$. This requires at least $r$ new edges (i.e. $|\mathcal E| \geq |\mathcal E_S|+2r$). We move $v_2 \in S_1$ to $S_2$. Then, since $S_1$ is $r$-reachable, a node $v_3 \in S_1$ has edges with $r$ nodes in $S_2$. This requires at least $r$ edges (i.e. $|\mathcal E| \geq |\mathcal E_S|+3r$). We then move $v_3 \in S_1$ to $S_2$.
    
    We can continue this process of (1) drawing edges between $v_i \in S_1$ and $r$ nodes in $S_2$ and (2) moving $v_i \in S_1$ into $S_2$ until $S_1$ becomes empty. Then, we have to do it for $r-1$ times, as we initially have $|S_1|=r-1$. In other words, we have $r-1$ different pairs of $S_1$ and $S_2$, and each pair requires $r$ new edges. Since we have started from the setting where $S_2= \mathcal V_S$, we have $|\mathcal E| \geq |\mathcal E_S| + r(r-1)$. Therefore, if $r$ is even, $|\mathcal E| \geq \frac {r^2+2} 2 +r(r-1)=\frac {r(3r-2)+2} 2$. If $r$ is odd, and $|\mathcal E| \geq \frac {r^2+1} 2 +r(r-1) = \frac {r(3r-2)+1} 2$.
\end{proof}

These lemmas and theorems give us insights into the structures of graphs with $2r-1$ and $2r$ nodes to maximize their robustness levels to $r$. Using these insights, we construct classes of $r$-robust graphs that maintain the maximum robustness with the least number of edges in the next section.

\section{Construction of Sparsest $r$-robust Graphs}
\label{sec:construction}
In this section, we use the lemmas and theorems from the previous section to introduce graphs of maximum robustness with minimal sets of edges. By doing so, we also show the bounds in Theorem~\ref{thm:min_addition} and~\ref{thm:min_addition2} are tight.

Let $n,r\in \mathbb Z_{>0}$ and let $\mathcal A$ be a set of all simple undirected graphs. Let $r(\mathcal G)$ be robustness of a graph $\mathcal G \in \mathcal A$. Let $\mathcal A_{n,r}$ be a subset of $\mathcal A$ that contains all $r$-robust graphs with $n$ nodes i.e. $\mathcal A_{n,r}=\{\mathcal G=(\mathcal V, \mathcal E) \in \mathcal A: r(\mathcal G)= r, |\mathcal V|=n\}$. We now define a class of graphs with the minimum number of edges among the set of $r$-robust graphs of $n$ nodes:

\begin{definition}[$\mathbf {(n,r)}$-$\textbf{robust Graph}$] Consider a graph $\mathcal G=(\mathcal V, \mathcal E) \in \mathcal A_{n,r}$. It is $\mathbf {(n,r)}$-\textbf{robust} if $|\mathcal{E}|\leq |\mathcal{E}_i|$ $\forall \mathcal G_i=(\mathcal V_i, \mathcal E_i)\in \mathcal A_{n,r}$, $i\in \{1,2,\cdots, |\mathcal A_{n,r}|-1, |\mathcal A_{n,r}|\}$.
\end{definition}

In this section, we limit ourselves to two special cases, namely $(2r-1,r)$-robust and $(2r,r)$-robust graphs, that also have maximum robustness.

For the first case, we show that $(2r-1,r)$-robust graph is a subclass of graphs called $F$-elemental graphs \cite{Guerrero17}, which we introduce below:

\begin{definition}[\textbf{$F$-elemental Graph}\cite{Guerrero17}]
    An $F$-elemental graph is a graph with $n=4F+1$ nodes that is $r$-robust with $r=2F+1$, $F\in \mathbb{Z}_{>0}$.
\end{definition}
\begin{prop}
\label{prop:felemental}
     Let $\mathcal G = (\mathcal V, \mathcal{E})$ be a graph where $|\mathcal V|=4F+1=2r-1$. Let $K\subset \mathcal V$ be a set of $r-1=2F$ nodes. If each node $i$ in $K$ connects to all nodes in $\mathcal V\setminus\{i\}$, and the remaining $r=2F+1$ nodes in $\mathcal V\setminus K$ form a connected subgraph, then $\mathcal G$ is $r$-robust \cite{Guerrero17}.
\end{prop}
\begin{remark}
    While Proposition~\ref{prop:felemental} specifically addresses odd values of $r$, the proof remains valid regardless of whether r is odd or even. Therefore, Proposition~\ref{prop:felemental} also extends to even values of $r$. For more detail, the reader is referred to \cite{Guerrero17}.
\end{remark}

Now we introduce how to construct $(2r-1,r)$-robust graphs in the following proposition:

\begin{prop}
\label{prop:2r-1}
 Let $\mathcal G = (\mathcal V, \mathcal{E})$ be a graph where $|\mathcal V|=2r-1$. Let $K\subset \mathcal V$ be a set of $r-1$ nodes. If each node $i$ in $K$ connects to all nodes in $\mathcal V\setminus\{i\}$, and the remaining $r$ nodes in $\mathcal V\setminus K$ form a tree graph, then $\mathcal G$ is $(2r-1,r)$-robust.
\end{prop}

\begin{proof}
To prove $\mathcal G$ is $(2r-1,r)$-robust, we need to prove two things: 1) it is $r$-robust and 2) $|\mathcal E|=\frac {3r(r-1)}2$ from Theorem~\ref{thm:min_addition}. 
We first prove that it is $r$-robust. By definition, $\mathcal G$ is an $F$-elemental graph with $F=\left\lfloor\frac {r-1} 2\right\rfloor$. Since we know that $F$-elemental graphs with $r=2F+1$ is $r$-robust from Proposition~\ref{prop:felemental}, we know that $\mathcal G$ is also $r$-robust.
Now we show that $\mathcal{E}$ is a minimal set. $\mathcal G$ has $\frac {(r-1)(r-2)} 2 + r(r-1)=\frac {(r-1)(3r-2)} 2$ edges for the nodes in $K$ connecting to every other node. Then, it further has $r-1$ edges as the remaining nodes form a tree graph. Adding them together, we get $\frac {3r(r-1)} 2$, which equals to \eqref{eq:bound1}.
\end{proof}

The difference between $F$-elemental graphs and $(2r-1,r)$-robust graphs is that the former includes the latter but not vice versa. Now, we investigate $(2r,r)$-robust graphs.

\begin{figure}
    \centering
    \includegraphics[width=5cm]{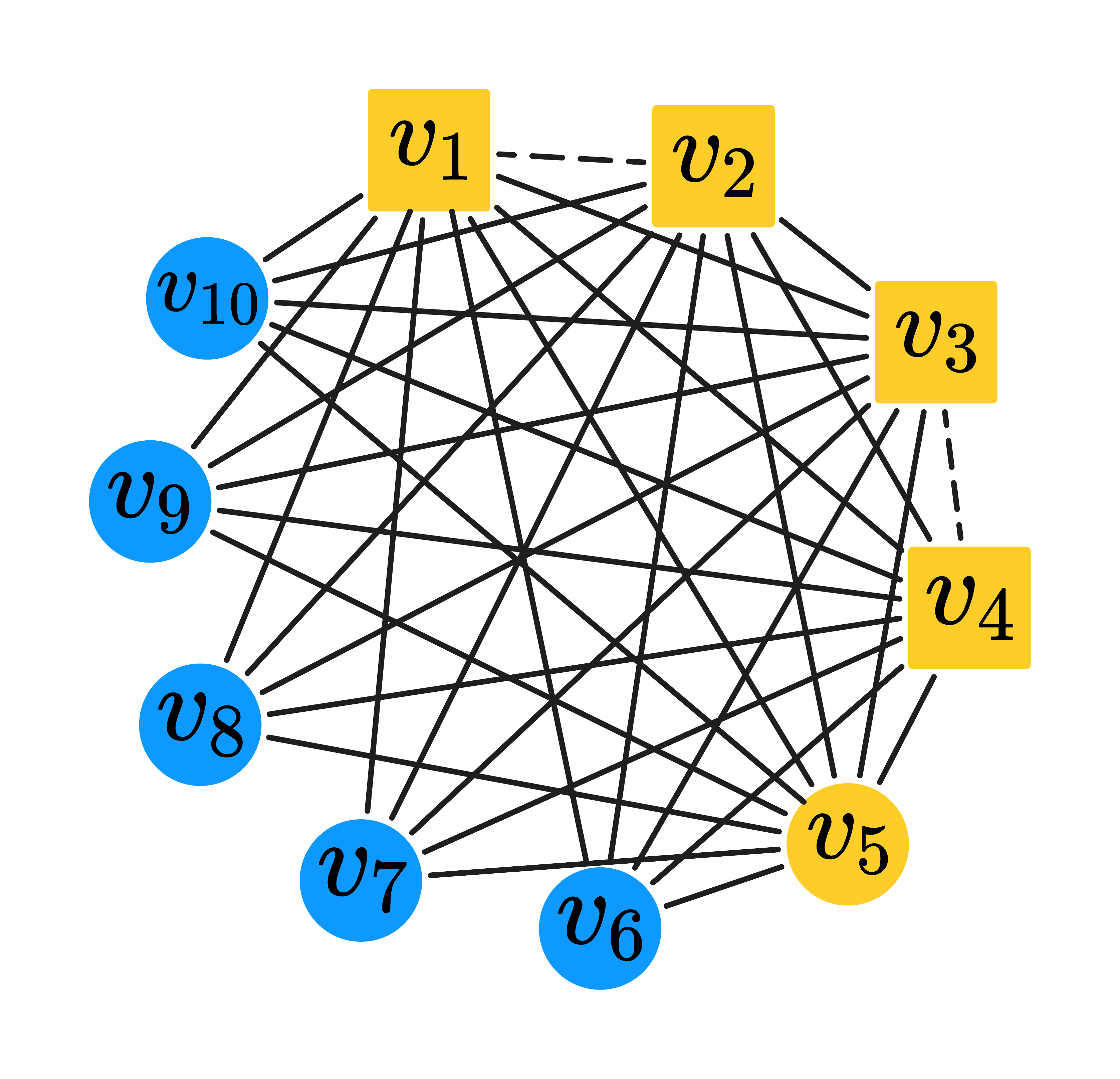}
    \caption{This figure visualizes a $(10,5)$-robust graph. Each of the $r=5$ nodes in $K = \{v_1,v_2,v_3,v_4,v_5\}$ (colored in yellow) connects to all $9$ nodes, and $\delta=4$ of them (in square) lose one edge which is represented in a dotted line. Since $\delta$ is always even ($4$ in this case), we can always form $\frac \delta 2$ disjoint pairs of two nodes from $\delta$ nodes ($2$ pairs i.e. $\{v_1,v_2\}$ and $\{v_3,v_4\}$ in this case). Then, we remove the edges between the pairs.}
    \label{fig:example}
    \vspace{-4mm}
\end{figure}

\begin{prop}
\label{prop:2r}
     Let $\mathcal G = (\mathcal V,\mathcal{E})$ be a graph with $|\mathcal V|=2r$. Let $K\subset \mathcal V$ be a set of $r$ nodes, where each node $i$ in $K$ connects to all nodes in $\mathcal V\setminus\{i\}$. Then, choose $\delta$ nodes in $K$ to form disjoint $\frac \delta 2$ pairs of nodes, where $\delta=r-1$ if $r$ is odd and $\delta=r-2$ if $r$ is even. If an edge between two nodes in each pair gets removed (see Fig.~\ref{fig:example}), $\mathcal G$ is $(2r,r)$-robust.
\end{prop}

\begin{proof}
To prove $\mathcal G$ is $(2r,r)$-robust, we need to prove two things: 1) it is $r$-robust and 2) $|\mathcal E|$ equals to \eqref{eq:bound2}. First, we examine its edge set. The $r$ nodes in $K$ connecting to every other node adds up to a total of $\frac {r(r-1)+2r(r)} 2 =\frac {3r^2-r} 2$ edges. Then, if we subtract $\frac \delta 2$ from it, we get $\frac {r(3r-2)+2} 2$ for even $r$ and $\frac {r(3r-2)+1} 2$ for odd $r$, which equals to \eqref{eq:bound2}.

    Now, we prove that $\mathcal G$ is $r$-robust. Let $K' = \mathcal V\setminus K$. Then, $|K|=|K'|=r$. WLOG, let $S_1$ and $S_2$ be nonempty subsets of $\mathcal V$ such that $S_1\cap S_2=\emptyset$ and $|S_1|\leq|S_2|$. That means $1\leq|S_1|\leq r$. There are two cases. (i) If $K \cap S_1 = \emptyset$, let a node $i \in K'\cap S_1$. Then, $|\mathcal V_i^N\setminus S_1|\geq r$, as $K \subseteq \mathcal V_i^N \setminus S_1$, making $S_1$ $r$-reachable. (ii) If $K \cap S_1 \neq \emptyset$, let a node $i \in K\cap S_1$. At worst case $|\mathcal V_i^N\setminus S_1| = |\mathcal V_i^N|-|S_1\cap \mathcal V_i^N|$ where $|\mathcal V_i^N|=2r-2$ and $1\leq |S_1\cap \mathcal V_i^N|\leq r-1$. This is the worst case, as there exists at least one node $x_1 \in K$ that has $2r-1$ neighbors instead of $2r-2$ (i.e. $|\mathcal V_{x_1}^N|=2r-1$). If $1\leq |S_1\cap \mathcal V_i^N|\leq r-2$, $|\mathcal V_i^N\setminus S_1| \geq r$, but if $|S_1\cap \mathcal V_i^N|=r-1$, $|\mathcal V_i^N\setminus S_1|=r-1$, which does not make $S_1$ $r$-reachable. However, note that $|S_1\cap \mathcal V_i^N|=r-1$ implies $|S_1|=|S_2|=r$. WLOG, let $x_1 \in S_1$. Then, $|\mathcal V_{x_1}^N\setminus S_1| =2r-1-|S_1\cap\mathcal V_{x_1}^N| \geq r$, since $1\leq|S_1\cap\mathcal V_{x_1}^N|\leq r-1$. Thus, either $S_1$ or $S_2$ is $r$-reachable even when $|S_1\cap \mathcal V_i^N|=r-1$. Thus, whether $S_1$ contains a node in $K$ or not, either $S_1$ or $S_2$ is always $r$-reachable.
\end{proof}

Proposition~\ref{prop:2r-1} and~\ref{prop:2r} show that the bounds presented at Theorem~\ref{thm:min_addition} and~\ref{thm:min_addition2} are tight. Note that the construction mechanisms shown at Proposition~\ref{prop:2r-1} and \ref{prop:2r} are not the only ways to construct $(2r-1,r)$-robust and $(2r,r)$-robust graphs. Also, note $\lceil{\frac {2r} 2}\rceil = \lceil{\frac {2r-1} 2}\rceil=r$. Thus, these graphs have the maximum $r$-robustness as well as the least number of edges.

 \begin{figure}
\includegraphics[width=0.9\columnwidth]{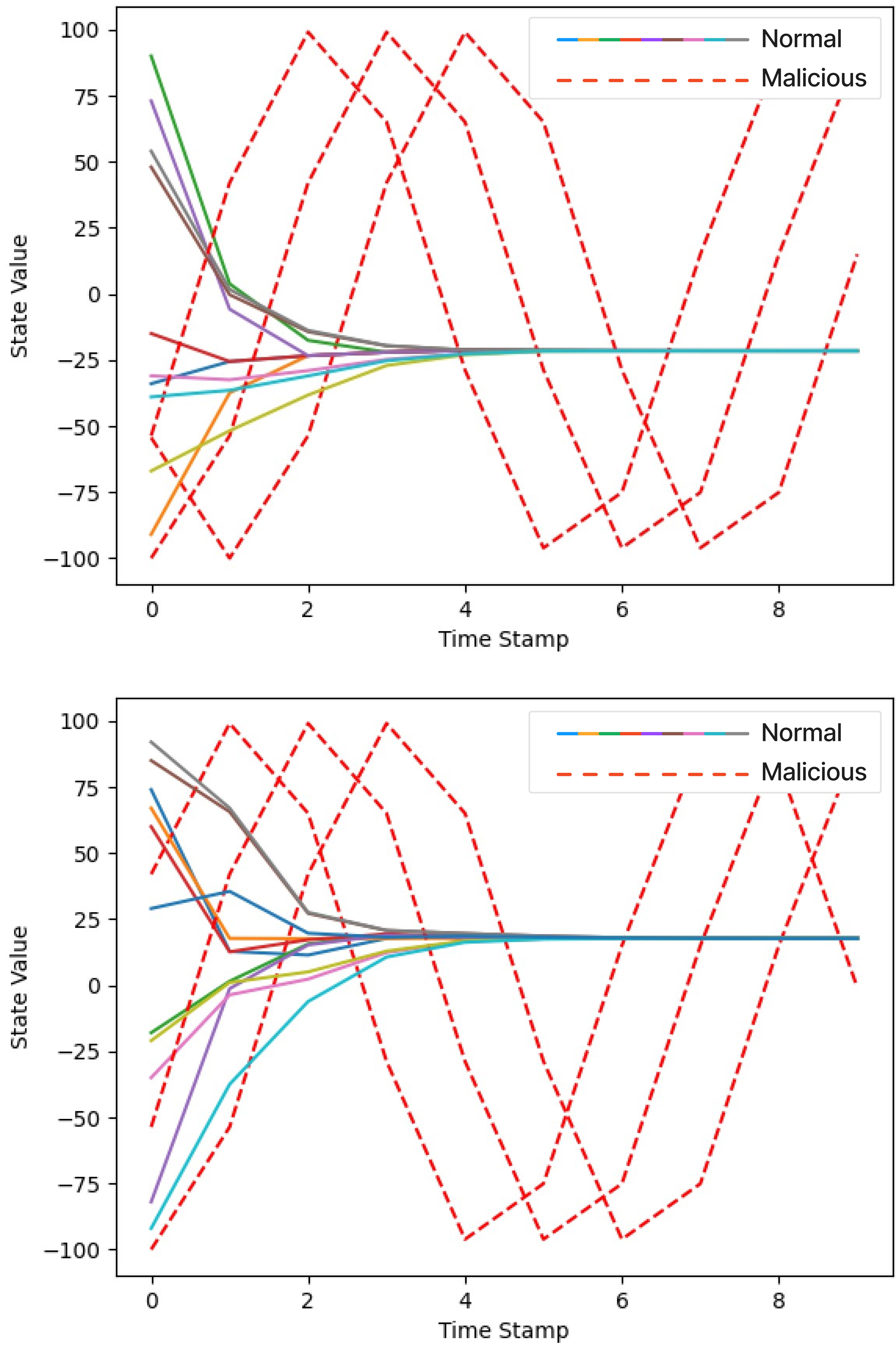}
    \caption{Simulations on the normal agents in the $(13,7)$-robust (top) and $(14,7)$-robust (bottom) graphs performing resilient consensus with $3$-local malicious agents through the W-MSR algorithm. The malicious agents' states are plotted in red dotted lines, while the normal agents' states are plotted in solid colored lines.}
    \label{fig:consensus1}
    \vspace{-4mm}
\end{figure}

\section{Simulations}
\label{sec:sim}
In previous section, we showed that $(2r-1,r)$-robust and $(2r,r)$-robust graphs have the least number of edges to reach maximum robustness. In this section, we present simulations to demonstrate their (a) maximum robustness and (b) properties of having the least number of edges among the $r$-robust graphs with the same number of nodes.

 In the first set of simulations, we have $13$ and $14$ agents in two separate scenarios. Then, the maximum robustness they can achieve is $\lceil{\frac {13} 2}\rceil=\lceil{\frac {14} 2}\rceil=7$. Thus by using procedures from Proposition~\ref{prop:2r-1} and~\ref{prop:2r}, the agents form $7$-robust graphs by constructing  $(13,7)$-robust and $(14,7)$-robust graphs. 

We choose malicious agents \cite{LeBlanc13} as our misbehaving agents' threat model. Malicious agents do not follow the nominal update protocol \eqref{eq:linear} but share the same values of their states with all of their neighbors at time $t$ \cite{LeBlanc13}. It is established in \cite{LeBlanc13} that normal agents in a $(2F+1)$-robust network are guaranteed to reach asymptotic consensus through the W-MSR algorithm in the presence of $F$-local malicious agents. Therefore by Proposition~\ref{prop:2r-1} and~\ref{prop:2r}, the normal agents in the $(13,7)$-robust and $(14,7)$-robust graphs are guaranteed to achieve consensus through the W-MSR algorithm with $3$-local malicious agents. The normal agents' initial states $x[0] \in \mathbb R$ are randomly generated on the interval $[-100,100]$. Fig.~\ref{fig:consensus1} shows successful consensus of the normal agents in $(13,7)$-robust and $(14,7)$-robust graphs. Their states are plotted in solid colored lines, while the malicious agents' states are plotted in red dotted lines.

In the second set of simulations, we have randomly generated $r$-robust graphs to empirically demonstrate that the bounds given in Theorem~\ref{thm:min_addition} and~\ref{thm:min_addition2} hold for any $r$. If they hold for any $r$, that implies $(2r-1,r)$-robust and $(2r,r)$-robust graphs have the least number of edges possible to maintain $r$-robustness among graphs of $2r-1$ and $2r$ nodes. 

 For each $r\in \{1,2,\cdots, 20\}$, we randomly generated $200$ Erdös–Rényi random graphs \cite{erdos1964,bollobas2001}. The graph version we adopted has two parameters - number of nodes, $n$, and the probability, $p$, of which an edge between two nodes is formed independent of other edges. In this simulation, we used $n\in \{2r-1,2r\}$ and $p\in\{0.7,0.75,0.8,0.85,0.9\}$. For each value of $r$, we randomly generated graphs until we had $200$ $r$-robust graphs of $n$ nodes ($50$ for each value of $p$), found the minimum number of edges, and compared them with the bounds. We computed the graphs' robustness with the method developed in \cite{usevitch2020}. Fig.~\ref{fig:sparse} compares the minimum number of edges among the $200$ graphs and the bounds presented in Theorem~\ref{thm:min_addition} and~\ref{thm:min_addition2} for each robustness. The green and red dotted lines represent the minimum number of edges among $200$ Erdös–Rényi random graphs with $2r-1$ and $2r$ nodes respectively. The blue and cyan solid lines represent the lower bounds from Theorem~\ref{thm:min_addition} and~\ref{thm:min_addition2} respectively. Note that the gaps between the bounds and found minimums increase as $r$ increases. This occurs because as $r$ increases, the number of different possible graphs grows exponentially, and thus it gets more difficult to generate graphs with lesser edges within the fixed number of $200$ generations.

  \begin{figure}
\centering
\includegraphics[width=0.98\columnwidth]{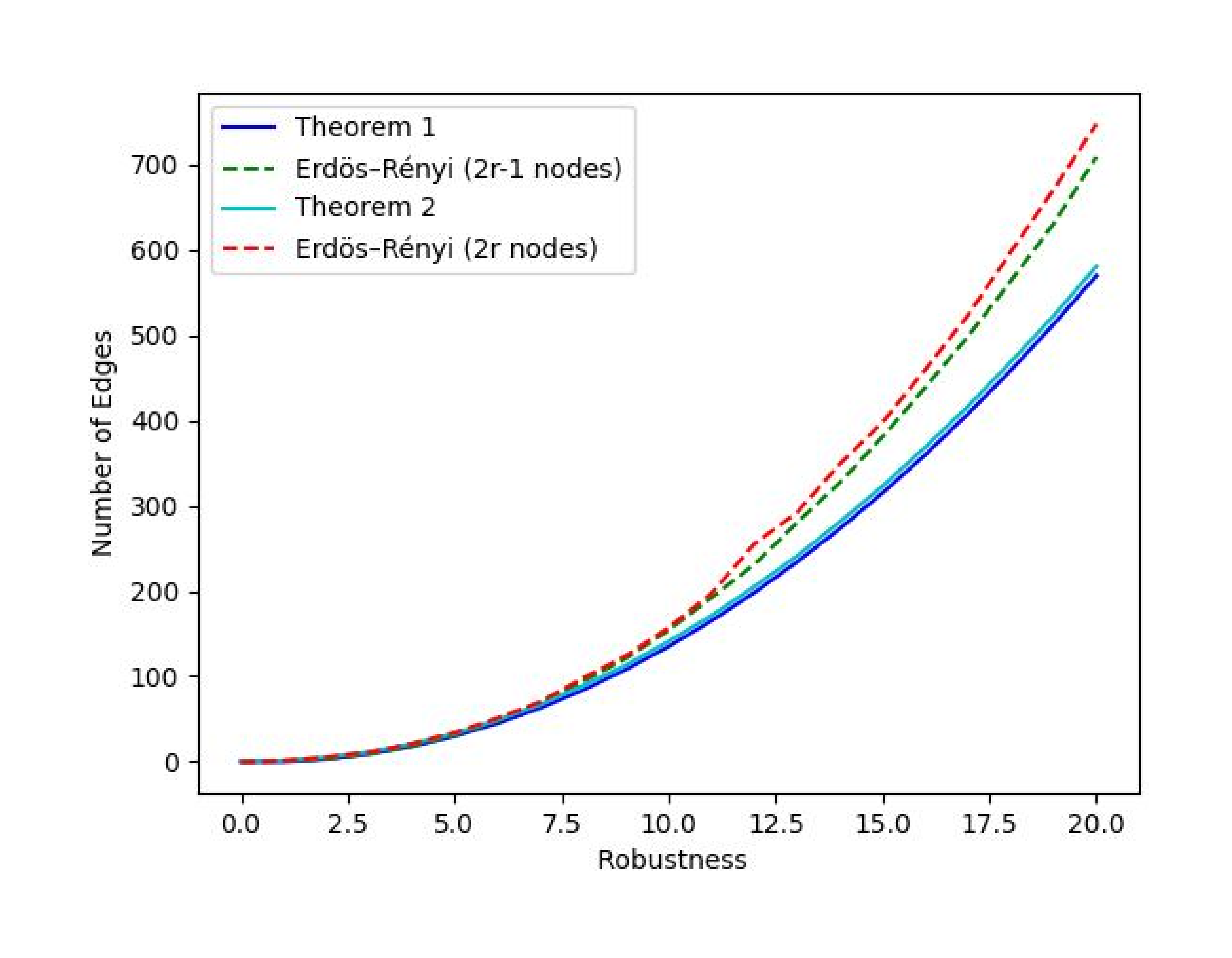}
    \caption{A comparison between the minimum numbers of edges among $200$ Erdös–Rényi random graphs of $2r-1$ and $2r$ nodes, and the lower bounds given at Theorem~\ref{thm:min_addition} and~\ref{thm:min_addition2} for each $r$.}
    \label{fig:sparse}
     \vspace{-4mm}
\end{figure}

\section{Conclusion}
\label{sec:concl}
In this paper, we study the properties of $r$-robust graphs in two ways. (a) We establish the necessary structures and tight lower bounds of number of edges for $r$-robust graphs with the maximum robustness. (b) Then we introduce $(2r-1,r)$-robust and $(2r,r)$-robust graphs, which are subclasses of a more general graph we call $(n,r)$-robust graphs, and prove that they exhibit the least number of edges to reach the maximum robustness with a given number of nodes. Finally, we empirically verify their robustness and sparsity properties through simulations. For future work, we aim to expand our work to $(n,r)$-robust graphs with $n>2r$, and study control synthesis that minimizes the loss of robustness as agents navigate challenging environments.

\bibliographystyle{IEEEtran}
\bibliography{references_ll}
\end{document}